\documentclass[12pt,twoside]{article}

 \usepackage{float}
\usepackage{graphicx}
\usepackage{epstopdf}
\usepackage{graphicx}
\usepackage{epic}
\usepackage{multirow}
\usepackage{tikz}
\usepackage{xcolor}
\usepackage{makecell}
\usepackage{threeparttable}
\usetikzlibrary{arrows,shapes,chains}

\renewcommand{\paragraph}{\roman{paragraph}}
\usepackage[a4paper]{geometry}
\makeatletter
\renewcommand\title[1]{\gdef\@title{\reset@font\Large\bfseries #1}}
\renewcommand\section{\@startsection {section}{1}{\z@}%
                                   {-3.5ex \@plus -1ex \@minus -.2ex}%
                                   {2.3ex \@plus.2ex}%
                                   {\normalfont\large\bfseries}}
\renewcommand\subsection{\@startsection{subsection}{2}{\z@}%
                                     {-3ex\@plus -1ex \@minus -.2ex}%
                                     {1.5ex \@plus .2ex}%
                                     {\normalfont\normalsize\bfseries}}
\renewcommand\subsubsection{\@startsection{subsubsection}{3}{\z@}%
                                     {-2.5ex\@plus -1ex \@minus -.2ex}%
                                     {1.5ex \@plus .2ex}%
                                     {\normalfont\normalsize\bfseries}}

\def\@runningauthor{}\newcommand{\runningauthor}[1]{\def\runningauthor{#1}}
\def\@runningtitle{}\newcommand{\runningtitle}[1]{\def\runningtitle{#1}}

\renewcommand{\ps@plain}{%
\renewcommand{\@evenhead}{\footnotesize\scshape \hfill\runningauthor\hfill}
\renewcommand{\@oddhead}{\footnotesize\scshape \hfill\runningtitle\hfill}}

\newcommand{\F}{\mathbb{F}}

\g@addto@macro\bfseries{\boldmath}

\makeatother

\usepackage{amsthm,amsmath,amssymb}
\usepackage{cite}
\usepackage{graphicx}

\usepackage[colorlinks=true,citecolor=black,linkcolor=black,urlcolor=blue]{hyperref}

\theoremstyle{plain}
\newtheorem{theorem}{Theorem}[section]

\newtheorem{lem}[theorem]{Lemma}
\newtheorem{cor}[theorem]{Corollary}
\newtheorem{prop}[theorem]{Proposition}

\theoremstyle{definition}

\newtheorem{example}[theorem]{Example}
\newtheorem{conjecture}[theorem]{Conjecture}

\theoremstyle{remark}
\newtheorem{remark}[theorem]{Remark}

\runningauthor{}

\date{}

\begin{document}
\begin{sloppypar}

\title{An open problem and a conjecture on binary linear complementary pairs of codes\thanks{The research of Shitao Li and Minjia Shi is supported by the National Natural Science Foundation of China under Grant 12071001. The research of San Ling is supported by Nanyang Technological University Research Grant 04INS000047C230GRT01.}}
\author{Shitao Li, Minjia Shi\thanks{Shitao Li and Minjia Shi are with the Key Laboratory of Intelligent Computing and Signal Processing, Ministry of Education, School of Mathematical Sciences, Anhui University, Hefei, 230601, China. They are also with the State Key Laboratory of Integrated Services Networks, Xidian University, Xi'an, 710071, China (email: lishitao0216@163.com, smjwcl.good@163.com).}, San Ling\thanks{San Ling is with the School of Physical and Mathematical Sciences, Nanyang Technological University, Singapore 637371, (email: lingsan@ntu.edu.sg).}}

\maketitle

\begin{abstract}
The existence of $q$-ary linear complementary pairs (LCPs) of codes with $q> 2$ has been completely characterized so far.
This paper gives a characterization for the existence of binary LCPs of codes. As a result, we solve an open problem proposed by Carlet $et~al.$ (IEEE Trans. Inf. Theory 65(3): 1694-1704, 2019) and a conjecture proposed by Choi $et~al.$ (Cryptogr. Commun. 15(2): 469-486, 2023).
\end{abstract}

\section{Introduction}
Let $q$ be a prime power and let $\F_q$ denote the finite field with $q$ elements. The {\em (Hamming) weight} ${\rm wt}({\bf x})$ of ${\bf x}=(x_1,x_2,\ldots,x_n)\in \F_q^n$ is defined by ${\rm wt}({\bf x})=|\{i~|~x_i\neq 0,~1\leq i\leq n\}|$. A {\em linear $[n,k,d]_q$ code} $C$ is a $k$-dimensional linear subspace of $\F_q^n$, where $d$ is the minimum weight of all nonzero codewords of $C$.
The {\em (Euclidean) dual} code $C^{\perp}$ of a linear $[n,k]_{q}$ code $C$ is defined by
$C^{\perp}=\{\textbf y\in \F_{q}^{n}~|~\langle \textbf x, \textbf y\rangle=0~ {\rm for\ all}\ \textbf x\in C \},$
where $\langle \textbf x, \textbf y\rangle=\sum_{i=1}^n x_iy_i$ for ${\bf x} = (x_1, x_2, \ldots, x_n)$ and $\textbf y = (y_1, y_2, \ldots, y_n)\in \F_{q}^{n}$.
A pair of linear codes $(C_1, C_2)$ of length $n$ is called a {\em linear $\ell$-intersection pair} of codes if $\dim(C_1 \cap C_2)=\ell$. This notion was introduced by Guenda $et~al.$ \cite{2020DCC-lIP}. In the special case when $\dim(C_1)+\dim(C_2)=n$ and $\ell=0$, $(C_1, C_2)$ is called a {\em linear complementary pair} (LCP) of codes. This notion was introduced by Carlet $et~al.$ \cite{LCP-IT-2018}. If $\dim(C_1)=k$, then $(C_1, C_2)$ is called an $[n,k]$ LCP of codes. In the more special case when $(C,C^\perp)$ is an LCP, then $C$ is called {\em linear complementary dual} (LCD). This notion was introduced by Massey \cite{LCD-Massey} in order to provide an optimum linear coding solution for the two-user binary adder channel.

LCPs of codes have been used in the framework of direct sum masking, which was proposed as a countermeasure against side channel attacks (SCAs) and fault injection attacks (FIAs) \cite{lcd-appl,LCP-Appl-1}. If an LCP of codes $(C_1, C_2)$ is used in this framework, then the minimum distance $d(C_1)$ of $C_1$ measures the protection against FIAs, whereas the minimum distance $d(C_2^\perp)$ of the dual code of $C_2$ measures the protection against SCAs. It has been shown that the level of resistance against both SCA and FIA depends on $\min\{d(C_1),d(C_2^\perp)\}$, which is referred to as the security parameter of the LCP of codes. Note that for an LCD code $C$, $(C,C^\perp)$ is an LCP of codes. Then the security parameter of $(C,C^\perp)$ is simply $d(C)$, which has been studied extensively (see \cite{AH-TLCD-1-10,AH-BLCD-17-24,AHS-BLCD,ter-11-19,LCD-lp,2-LCD-30,HS-BLCD-1-16,2-LCD-40,B-LCD-40,bound-12,234-lcD,LCD-equivalent,IT-Jin}). Hence we can define the optimal security parameter of $[n,k]$ LCPs of codes for given $n$ and $k$ as follows:
$$d_{LCP}(n,k)=\max\{\min\{d(C_1),d(C_2^\perp)\}~|~(C_1,C_2)~{\rm is~an}~[n,k]_q~{\rm LCP~of~codes}\}.$$
Let $d_L(n, k)$ denote the largest minimum distance among linear $[n, k]_q$ codes. Then we clearly have $d_{LCP}(n,k)\leq d_L(n,k).$

Carlet $et~al.$ \cite{Sigam-LCD} showed that for any $q$-ary linear code $C$ with $q>2$, there is a monomial transformation $\sigma$ on $\F_q^n$ such that $(C,(\sigma(C))^\perp)$ is an LCP of codes. They proposed an open problem for the existence of $q$-ary LCPs of codes: is there a monomial transformation $\sigma$ on $\F_q^n$ such that $(C_1,(\sigma(C_2))^\perp)$ is an LCP of codes for any two $q$-ary linear $[n,k]$ codes $C_1$ and $C_2$ with $q>2$? This open problem has been settled by Anderson $et~al.$ \cite{Relative-hulls}. This means that the problems of the existence and the optimal security parameters of $q$-ary LCPs of codes have been settled for $q>2$.
For the binary case, the same article \cite{Sigam-LCD} showed that for any binary linear code $C$, there is a linear transformation $\sigma$ on $\F_2^{n+1}$ such that $(\{0\}\times C,(\sigma(\{0\}\times C))^\perp)$ is an LCP of codes. They also proposed a similar open problem for the existence of binary LCPs of codes.
Moreover, their result also yields that $d_{LCP}(n,k)\geq d_L(n,k)-1$. Hence an interesting topic is to study the best security parameter problem for binary LCPs of codes. For this topic, Choi $et~al.$ \cite{2023CCDS-LCP} determined the best security parameters of binary $[n,k]$ LCPs of codes for $n\leq 18$ or $k\leq 4$. Further, they also described a sufficient condition for $d_{LCP}(n,k)= d_L(n,k)-1$, and gave a conjecture on its necessary condition. Very recently, G${\rm\ddot{u}}$neri \cite{2023ITSSLCP} constructed an infinite family of optimal binary LCPs of codes from Solomon-Stiffler codes.

In this paper, we study the problems of the existence and the optimal security parameters of binary LCPs of codes. We show that for two binary linear $[n,k]$ codes $C_1$ and $C_2$, there is a linear transformation $\sigma$ on $\F_2^{n+1}$ such that $(\{0\}\times C_1,(\sigma(\{0\}\times C_2))^\perp)$ is an LCP of codes. This solves an open problem proposed by Carlet $et~al.$ \cite{Sigam-LCD}. Further, we present a more general sufficient condition for $d_{LCP}(n,k)= d_L(n,k)-1$ than that given by Choi $et~al.$ \cite{2023CCDS-LCP}, and show a conjecture proposed in the same paper. As a result, we construct an infinite family of optimal binary LCPs of codes from Solomon-Stiffler codes, whose parameters cover those constructed by G${\rm\ddot{u}}$neri \cite{2023ITSSLCP}.
Further, these results also imply that the conjecture proposed by Guenda $et~al.$ \cite{2020DCC-lIP} is not valid for the binary case.

The paper is organized as follows. The next section gives some notations and preliminaries. In Section 3, we show that there is a linear transformation $\sigma$ on $\F_2^{n+1}$ such that $(\{0\}\times C_1,(\sigma(\{0\}\times C_2))^\perp)$ is an LCP of codes for two binary linear $[n,k]$ codes $C_1$ and $C_2$. In Section 4, we present a characterization for $d_{LCP}(n,k)=d_L(n,k)-1$. In Section 5, we conclude the paper.

\section{Preliminaries}
Throughout this paper, let $O$ denote an appropriate zero matrix and let ${\bf0}$ denote an appropriate zero row (resp. column) vector.
A {\em generator matrix} of a binary linear $[n, k]$ code $C$ is any $k\times n$ matrix $G$ whose rows form a basis of $C$. A {\em parity check matrix} of a binary linear $[n, k]$ code $C$ is any generator matrix of $C^\perp$. A binary vector ${\bf x} = (x_1, x_2, \ldots , x_n)$ is {\em even-like} if $\sum_{i=1}^nx_i=0$ and is {\em odd-like} otherwise. A binary linear code is said to be {\em even-like} if it has only even-like codewords, and is said to be {\em odd-like} if it
is not even-like. The Griesmer bound \cite{Huffman} is defined on a binary linear $[n,k,d]$ code $C$ as
$$n\geq g(k,d):=\sum_{i=0}^{k-1}\left\lceil \frac{d}{2^i}\right\rceil,$$
 where $\lceil a \rceil$ is the least integer greater than or equal to the real number $a$.
A binary linear $[n, k, d]$ code $C$ is said to be a {\em Griesmer code} if $n$ meets the Griesmer bound, $i.e.$, $n=g(k,d)$. Assume that $S_k$ is a matrix whose columns are all the nonzero vectors in $\F_2^k$.
It is well-known that $S_k$ generates a binary simplex code, which is a one-weight $[2^k-1,k,2^{k-1}]$ Griesmer code for $k\geq 3$ (see \cite{Huffman}).
Consider the following matrix
$$R(1,k)=\begin{pmatrix}
    1 & 1\cdots1  \\
    {\bf0} & S_k
  \end{pmatrix},$$
which generates the {\em first order Reed-Muller code} $\mathcal{R}(1,k)$ (see \cite{Huffman}).
It can be checked that $\mathcal{R}(1,k)$ is a binary linear $[2^k,k+1,2^{k-1}]$ code containing the all-one vector ${\bf 1}$.

A {\em permutation} $\pi$ is an arrangement of $n$ elements without repetition, and the corresponding permutation matrix is denoted by $P_{\pi}$.
The {\em permutation group} {\bf PAut}$(\F_2^n)$ of $\F_2^n$ is the set consisting of all permutations. Let $\pi\in {\rm {\bf PAut}}(\F_2^n)$ and let $G$ be a $k\times n$ matrix. We use $\pi(G)$ to denote the $k\times n$ matrix with its $i$th row $\pi(G(i,:))$, where $G(i,:)$ is the $i$th row of $G$ for $1\leq i\leq k.$
For a $k\times n$ matrix $G$, we use ${\rm Row}_i(G)$ and ${\rm Col}_j(G)$ to denote the $i$th row and $j$th column of $G$, respectively.
Let $A$ and $B$ be two $k\times n$ matrices. Let $\pi\in {\rm {\bf PAut}}(\F_2^{n})$ be the permutation that transposes positions $i$ and $j$. Then it can be checked that
$$\pi(A)B=AP_{\pi}B=AB+({\rm Col}_j(A)-{\rm Col}_i(A))({\rm Row}_i(B)-{\rm Row}_j(B)).$$

Let $C$ be a binary linear $[n,k]$ code and let $\pi\in {\rm {\bf PAut}}(\F_2^n)$. Then it can be checked that $(\pi(C))^\perp=\pi(C^\perp)$ by a simple argument.

\section{An open problem on binary LCPs of codes}
In this section, we present a characterization for the existence of binary LCPs of codes. As a result, we solve an open problem proposed by Carlet $et~al.$ \cite{Sigam-LCD}. First, we recall a useful proposition.

\begin{prop}{\rm\cite[Theorem 2.1]{2020DCC-lIP}}\label{prop-LIP-rank}
For $i=1,2$, let $C_i$ be a linear $[n,k_i]_q$ code with generator matrix $G_i$ and parity check matrix $H_i$. If $C_1$ and $C_2$ form an $\ell$-intersection pair, then
$${\rm rank}(G_1H_2^T)={\rm rank}(H_2 G_1^T)=k_1-\ell~{\rm and}~{\rm rank}(G_2H_1^T)={\rm rank}(H_1 G_2^T)=k_2-\ell.$$
\end{prop}

The following corollary can also be derived from {\rm\cite[Proposition 2.2]{Relative-hulls}}.

\begin{cor}\label{cor-rank}
For $i=1,2$, let $C_i$ be a binary linear $[n,k]$ code. Then $\dim(C_1\cap C_2^{\perp})=\dim(C_2\cap C_1^{\perp})$.
\end{cor}

\begin{proof}
For $i=1,2$, let $G_i$ be a generator matrix of $C_i$.
By Proposition \ref{prop-LIP-rank}, we have
$\dim(C_1\cap C_2^{\perp})=k-{\rm rank}(G_1G_2^T)~{\rm and}~\dim(C_2\cap C_1^{\perp})=k-{\rm rank}(G_2G_1^T).$
This yields that $\dim(C_1\cap C_2^{\perp})=\dim(C_2\cap C_1^{\perp})$ since ${\rm rank}(G_1G_2^T)={\rm rank}(G_2G_1^T)$.
\end{proof}

\begin{lem}\label{lem-LIP-1}
For $i=1,2$, let $C_i$ be a binary linear $[n,k]$ code. If $\dim(C_1\cap C_2^{\perp})=h\geq 1$, then for $1\leq \ell\leq h$, there exists a mapping $\sigma\in {\rm {\bf PAut}}(\F_2^{n})$ such that $(C_1,(\sigma(C_2))^\perp)$ is an $\ell$-intersection pair of codes, $i.e.$, $\dim(C_1\cap (\sigma(C_2))^\perp)=\ell$.
\end{lem}

\begin{proof}
If $h=1$ or $\ell=h$, then the result holds by choosing $\sigma$ as the identity map. Next, it is sufficient to consider the case where $h\geq 2$ and $1\leq \ell\leq h-1$.
By Corollary \ref{cor-rank}, $\dim(C_1\cap C_2^{\perp})=\dim(C_2\cap C_1^{\perp})=h.$
According to \cite[Lemma 3]{Sigam-LCD}, there exists a mapping $\pi_1\in {\rm {\bf PAut}}(\F_2^n)$ such that $\pi_1(C_2)$ has a generator matrix of
the form
$$G_2=\begin{pmatrix}
              I_h & A_2 \\
              O & B_2
            \end{pmatrix},$$
where $(I_h ~ A_2)$ is a generator matrix of $(\pi_1(C_1))^\perp\cap \pi_1(C_2)$. We can assume that $\pi(C_1)$ has a generator matrix of
the form
$$G_1=\begin{pmatrix}
              A_1 & B_1 \\
              D_1 & E_1
            \end{pmatrix},$$
where $(A_1 ~ B_1)$ is a generator matrix of $\pi_1(C_1)\cap (\pi_1(C_2))^\perp$ and $A_1$ is an $h\times h$ matrix. Hence $(I_h ~ A_2)G_1^T=O$ and $G_2(A_1 ~ B_1)^T=O$.
Then it is easy to see that
\begin{align*}
  G_2G_1^T & = \begin{pmatrix}
              I_h & A_2 \\
              O & B_2
            \end{pmatrix}
\begin{pmatrix}
              A_1 & B_1 \\
              D_1 & E_1
            \end{pmatrix}^T\\
   & =\begin{pmatrix}
              A_1^T+A_2B_1^T & D_1^T+A_2E_1^T \\
              B_2B_1^T & B_2E_1^T
            \end{pmatrix}\\
& =\begin{pmatrix}
              O & O \\
              O & B_2E_1^T
            \end{pmatrix},
\end{align*}
where $B_2E_1^T$ is a $(k-h)\times (k-h)$ matrix. By Proposition \ref{prop-LIP-rank}, we have ${\rm rank}(G_2G_1^T)={\rm rank}(B_2E_1^T)=k-h.$

\begin{itemize}
  \item {\bf Case 1.} Suppose that there exist $1\leq i<j\leq h$ such that ${\rm Col}_i(A_1)\neq {\rm Col}_j(A_1)$. Let ${\bf a}=({\rm Col}_{j}(G_2)-{\rm Col}_i(G_2))=(a_1,a_2,\ldots,a_k)^T$ and let ${\bf b}=({\rm Col}_{i}(G_1)-{\rm Col}_j(G_1))^T=(b_1,b_2,\ldots,b_k)$. Then there exists $1\leq i_0\leq h$ such that $b_{i_0}\neq 0$. Let $\pi_2\in {\rm {\bf PAut}}(\F_2^{n})$ be the permutation that transposes positions $i$ and $j$. Then
  \begin{align*}
\pi_2(G_2)G_1^T & = G_2G_1^T+({\rm Col}_{j}(G_2)-{\rm Col}_i(G_2))({\rm Row}_{i}(G_1^T)-{\rm Row}_j(G_1^T))\\
& = G_2G_1^T+({\rm Col}_{j}(G_2)-{\rm Col}_i(G_2))({\rm Col}_{i}(G_1)-{\rm Col}_j(G_1))^T\\
&=\begin{pmatrix}
              O & O \\
              O & B_2E_1^T
            \end{pmatrix}
+\begin{pmatrix}
                a_1{\bf b}\\
                a_2{\bf b}\\
                \vdots\\
                a_k{\bf b}
              \end{pmatrix}.
\end{align*}
Since $a_i=1$ and $1\leq i\leq h$, after row operations, we have the following row transformations
$$\pi_2(G_2)G_1^T\sim \begin{pmatrix}
              O & O \\
              O & B_2E_1^T
            \end{pmatrix}
+\begin{pmatrix}
                {\bf b}\\
                {\bf 0}\\
                \vdots\\
                {\bf 0}
              \end{pmatrix}.$$
Since $b_{i_0}\neq 0$ and $1\leq i_0\leq h$, ${\rm rank}(\pi_2(G_2)G_1^T)={\rm rank}(B_2E_1^T)+1=k-h+1$. By Proposition \ref{prop-LIP-rank}, we have $\dim(\pi_1(C_1)\cap (\pi_2\pi_1(C_2))^\perp)=h-1$. This implies that $\dim(C_1\cap (\pi_1^{-1}\pi_2\pi_1(C_2))^\perp)=h-1$.
  \item {\bf Case 2.} Suppose that ${\rm Col}_i(A_1)= {\rm Col}_j(A_1)={\bf d}$ for any $1\leq i<j\leq h$.

  {\bf Subcase 2.1.} If $A_1=O$, then ${\rm rank}(B_1)=h$. Hence there exists $1\leq j\leq n-h$ such that ${\rm Col}_j(B_1)$ is not the zero column. It is easy to see that there exists $1\leq i\leq 2$ such that ${\rm Col}_i(I_h)\neq {\rm Col}_{j}(A_2)$. Let ${\bf a}=({\rm Col}_{j+h}(G_2)-{\rm Col}_i(G_2))=(a_1,a_2,\ldots,a_k)^T$ and let ${\bf b}=({\rm Col}_{i}(G_1)-{\rm Col}_{j+h}(G_1))^T=(b_1,b_2,\ldots,b_k)$. Then there exist $1\leq i_0,j_0\leq h$ such that $a_{i_0}\neq 0$ and $b_{j_0}\neq 0$. Let $\pi_2\in {\rm {\bf PAut}}(\F_2^{n})$ be the permutation that transposes positions $i$ and $j+h$. Then
  \begin{align*}
\pi_2(G_2)G_1^T & = G_2G_1^T+({\rm Col}_{j+h}(G_2)-{\rm Col}_i(G_2))({\rm Row}_{i}(G_1^T)-{\rm Row}_{j+h}(G_1^T))\\
& = G_2G_1^T+({\rm Col}_{j+h}(G_2)-{\rm Col}_i(G_2))({\rm Col}_{i}(G_1)-{\rm Col}_{j+h}(G_1))^T\\
&=\begin{pmatrix}
              O & O \\
              O & B_2E_1^T
            \end{pmatrix}
+\begin{pmatrix}
                a_1{\bf b}\\
                a_2{\bf b}\\
                \vdots\\
                a_k{\bf b}
              \end{pmatrix}.
\end{align*}
Since $a_{i_0}\neq 0$ and $1\leq i_0\leq h$, after row operations, we have the following row transformations
$$\pi_2(G_2)G_1^T\sim \begin{pmatrix}
              O & O \\
              O & B_2E_1^T
            \end{pmatrix}+\begin{pmatrix}
                {\bf b}\\
                {\bf 0}\\
                \vdots\\
                {\bf 0}
              \end{pmatrix}.$$
Since $b_{j_0}\neq 0$ and $1\leq j_0\leq h$, ${\rm rank}(\pi_2(G_2)G_1^T)={\rm rank}(B_2E_1^T)+1=k-h+1$. By Proposition \ref{prop-LIP-rank}, we have $\dim(\pi_1(C_1)\cap (\pi_2\pi_1(C_2))^\perp)=h-1$. This implies that $\dim(C_1\cap (\pi_1^{-1}\pi_2\pi_1(C_2))^\perp)=h-1$.

{\bf Subcase 2.2.} If $A_1\neq O$, then this yields that ${\rm rank}(A_1)=1$. Since ${\rm rank}((A_1,B_1))= h\geq 2$, there exists $1\leq j\leq n-h$ such that ${\rm Col}_j(B_1)\neq {\bf d}$. It is easy to see that there exists $1\leq i\leq 2$ such that ${\rm Col}_i(I_h)\neq {\rm Col}_{j}(A_2)$. Let ${\bf a}=({\rm Col}_{j+h}(G_2)-{\rm Col}_i(G_2))=(a_1,a_2,\ldots,a_k)^T$ and let ${\bf b}=({\rm Col}_{i}(G_1)-{\rm Col}_{j+h}(G_1))^T$. Then there exist $1\leq i_0,j_0\leq h$ such that $a_{i_0}\neq 0$ and $b_{j_0}\neq 0$. Let $\pi_2\in {\rm {\bf PAut}}(\F_2^{n})$ be the permutation that transposes positions $i$ and $j+h$. Similar to {\bf Subcase 2.1}, we have ${\rm rank}(\pi_2(G_2)G_1^T)={\rm rank}(B_2E_1^T)+1=k-h+1$. By Proposition \ref{prop-LIP-rank}, we have $\dim(\pi_1(C_1)\cap (\pi_2\pi_1(C_2))^\perp)=h-1$. This implies that $\dim(C_1\cap (\pi_1^{-1}\pi_2\pi_1(C_2))^\perp)=h-1$.
\end{itemize}
In short, there exists a mapping $\sigma_1\in {\rm {\bf PAut}}(\F_2^{n})$ such that $\dim(C_1\cap (\sigma_1(C_2))^\perp)=h-1$. Similarly, there exists a mapping $\sigma_2\in {\rm {\bf PAut}}(\F_2^{n})$ such that $\dim(C_1\cap (\sigma_2\sigma_1(C_2))^\perp)=h-2$. Continuing this process, then there exists a mapping $\sigma_i\in {\rm {\bf PAut}}(\F_2^{n})$ such that $\dim(C_1\cap (\sigma_i\cdots \sigma_1(C_2))^\perp)=h-i$ for $1\leq i\leq h-1$.
For $1\leq \ell\leq h-1$, by taking $\sigma=\sigma_{h-\ell}\cdots \sigma_1\in {\rm {\bf PAut}}(\F_2^{n})$, we have $\dim(C_1\cap (\sigma(C_2))^\perp)=\ell$.
\end{proof}

Next, we give a characterization for the existence of binary LCPs of codes.

\begin{theorem}\label{thm-LlIP}
Let $C_1$ and $C_2$ be binary linear $[n,k]$ codes. Then there exists a mapping $\sigma\in {\rm {\bf PAut}}(\F_2^{n})$ such that $(C_1,(\sigma(C_2))^\perp)$ is an LCP of codes if and only if ${\bf 1}\notin C_1\cap C_2^\perp$ and ${\bf 1}\notin C_2\cap C_1^\perp$.
\end{theorem}

\begin{proof}
Suppose that $\dim(C_1\cap C_2^\perp)=h$. If $h=0$, then the result obviously holds by choosing $\sigma$ as the identity map. Next, it is sufficient to consider the case where $h\geq 1$.

$(\Longrightarrow)$ If ${\bf 1}\in C_1\cap C_2^\perp$, then for any $\sigma\in {\rm {\bf PAut}}(\F_2^n)$, we have ${\bf 1}\in \sigma(C_2^\perp)=(\sigma(C_2))^\perp.$ This implies that ${\bf 1}\in C_1\cap (\sigma(C_2))^\perp$. Hence $(C_1,(\sigma(C_2))^\perp)$ is not LCP for any $\sigma\in {\rm {\bf PAut}}(\F_2^n)$, which is a contradiction.
If ${\bf 1}\in C_2\cap C_1^\perp$, then ${\bf 1}\in \sigma(C_2)\cap C_1^\perp$ for any $\sigma\in {\rm {\bf PAut}}(\F_2^n)$ by a similar way. By Proposition \ref{prop-LIP-rank}, $\dim(C_1\cap (\sigma(C_2))^\perp)=\dim(\sigma(C_2)\cap C_1^\perp)\geq 1$. Hence $(C_1,(\sigma(C_2))^\perp)$ is not LCP for any $\sigma\in {\rm {\bf PAut}}(\F_2^n)$, which is a contradiction.

$(\Longleftarrow)$ Assume that ${\bf 1}\notin C_1\cap C_2^\perp$ and ${\bf 1}\notin C_2\cap C_1^\perp$. By Lemma \ref{lem-LIP-1}, there exists a mapping $\pi_1\in {\rm {\bf PAut}}(\F_2^{n})$ such that $\dim(C_1\cap \pi_1(C_2)^\perp)=1$. Assume that $C_1\cap (\pi_1(C_2))^\perp=\{{\bf 0},{\bf a}\}$ and $\pi_1(C_2)\cap C_1^\perp=\{{\bf 0},{\bf b}\},$ where ${\bf a}=(a_1,a_2,\ldots,a_n)$ and ${\bf b}=(b_1,b_2,\ldots,b_n)$. Based on assumption, we have ${\bf a}\neq {\bf 1}$ and ${\bf b}\neq {\bf 1}$.
We claim that there exist $1\leq i<j\leq n$ such that $a_i\neq a_j$ and $b_i\neq b_j$. Otherwise, for any $1\leq r<s\leq n$, $a_r\neq a_s$ must yield that $b_r=b_s$. Together with the facts that ${\bf a}\neq {\bf 0}$ and ${\bf a}\neq {\bf 1}$, we have ${\bf b}={\bf 0}$ or ${\bf 1}$. This contradicts the facts that ${\bf b}\neq {\bf 0}$ and ${\bf b}\neq {\bf 1}$. Therefore, the claim is valid.
We assume that $C_1$ and $\pi_1(C_2)$ respectively have the following generator matrices:
$$G_1=\begin{pmatrix}
              {\bf a} \\
              A
            \end{pmatrix}~{\rm and}~G_2=\begin{pmatrix}
              {\bf b} \\
              B
            \end{pmatrix}.$$
Hence $$G_2G_1^T=\begin{pmatrix}
              {\bf a} \\
              A
            \end{pmatrix}
\begin{pmatrix}
              {\bf b} \\
              B
            \end{pmatrix}^T=\begin{pmatrix}
              0 &{\bf 0}\\
              {\bf 0}&AB^T
            \end{pmatrix}.$$
This yields that ${\rm rank}(G_2G_1^T)={\rm rank}(AB^T)=k-1$. Let ${\bf c}=({\rm Col}_{j}(G_2)-{\rm Col}_i(G_2))=(c_1,c_2,\ldots,c_k)^T$ and let ${\bf d}=({\rm Col}_{i}(G_1)-{\rm Col}_j(G_1))^T=(d_1,d_2,\ldots,d_k)$. Let $\pi_2\in {\rm {\bf PAut}}(\F_2^{n})$ be the permutation that transposes positions $i$ and $j$. Then
      \begin{align*}
        \pi_2(G_2)G_1^T& = G_2G_1^T+({\rm Col}_{j}(G_2)-{\rm Col}_i(G_2))({\rm Row}_{i}(G_1^T)-{\rm Row}_j(G_1^T))\\
        & =G_2G_1^T+({\rm Col}_{j}(G_2)-{\rm Col}_i(G_2))({\rm Col}_{i}(G_1)-{\rm Col}_j(G_1))^T \\
        & =\begin{pmatrix}
              0 &{\bf 0}\\
              {\bf 0}&AB^T
            \end{pmatrix}+\begin{pmatrix}
                c_1{\bf d}\\
                c_2{\bf d}\\
                \vdots\\
                c_k{\bf d}
              \end{pmatrix}.
      \end{align*}
    Since $c_1=b_j-b_i\neq 0$, after row operations, we have the following row transformations
$$\pi_2(G_2)G_1^T\sim \begin{pmatrix}
              0 &{\bf 0}\\
              {\bf 0}&AB^T
            \end{pmatrix}+\begin{pmatrix}
                {\bf d}\\
                {\bf 0}\\
                \vdots\\
                {\bf 0}
              \end{pmatrix}.$$
Note that $d_1=a_i-a_j\neq0$. Hence ${\rm rank}(\pi_2(G_2)G_1^T)={\rm rank}(AB^T)+1=k$. By Proposition \ref{prop-LIP-rank}, we have $\dim(C_1\cap (\pi_2\pi_1(C_2))^\perp)=0$.
Let $\sigma=\pi_2\pi_1\in {\rm {\bf PAut}}(\F_2^n)$. Then $(C_1,(\sigma(C_2))^\perp)$ is an LCP of codes, $i.e.$, $\dim(C_1\cap (\sigma(C_2))^\perp)=0$.
\end{proof}

\begin{example}
Let $C_1$ and $C_2$ be two binary linear $[8,4]$ codes with the following generator matrices respectively,
\begin{center}
$G_1=\begin{pmatrix}
1&0&0&0&1&1&1&0\\
0&1&0&1&0&0&0&1\\
0&0&1&0&1&1&0&0\\
0&0&0&1&1&1&0&0
           \end{pmatrix}$ and $G_2=\begin{pmatrix}
1&0&0&0&0&0&1&0\\
0&1&0&0&0&0&0&1\\
0&0&1&1&0&1&1&1\\
0&0&0&1&1&1&0&1
           \end{pmatrix}.$
\end{center}
It can be checked that $C_1\cap C_2^\perp$ and $C_2\cap C_1^\perp$ are binary linear $[8,3]$ codes generated by the first three rows of $G_1$ and $G_2$, respectively. By Lemma \ref{lem-LIP-1}, let $\pi_1\in {\rm {\bf PAut}}(\F_2^8)$ be defined by $\pi_1(c_1,c_2,c_3,c_4,c_5,c_6,c_7,c_8)=(c_1,c_3,c_2,c_4,c_5,c_6,c_7,c_8).$ Then $C_1$ and $\pi_1(C_2)$ respectively have generator matrices as follows:
\begin{center}
$G'_1=\begin{pmatrix}
1&0&0&0&1&1&1&0\\
0&1&1&1&1&1&0&1\\
0&0&1&0&1&1&0&0\\
0&0&0&1&1&1&0&0
           \end{pmatrix}$ and $G'_2=\begin{pmatrix}
1&0&0&0&0&0&1&0\\
0&1&1&1&0&1&1&0\\
0&0&1&1&0&1&1&1\\
0&0&0&1&1&1&0&1
           \end{pmatrix}.$
\end{center}
It can be checked that $C_1\cap (\pi_1(C_2))^\perp$ and $\pi_1(C_2)\cap C_1^\perp$ are binary linear $[8,2]$ codes generated by the first two rows of $G'_1$ and $G'_2$, respectively.
By Lemma \ref{lem-LIP-1}, let $\pi_2\in {\rm {\bf PAut}}(\F_2^8)$ be defined by $\pi_1(c_1,c_2,c_3,c_4,c_5,c_6,c_7,c_8)=(c_2,c_1,c_3,c_4,c_5,c_6,c_7,c_8).$ Then $C_1$ and $\pi_2\pi_1(C_2)$ respectively have generator matrices as follows:
\begin{center}
$G''_1=\begin{pmatrix}
1&1&1&1&0&0&1&1\\
0&1&1&1&1&1&0&1\\
0&0&1&0&1&1&0&0\\
0&0&0&1&1&1&0&0
           \end{pmatrix}$ and $G''_2=\begin{pmatrix}
1&1&1&1&0&1&0&0\\
0&1&1&1&0&1&1&0\\
0&0&1&1&0&1&1&1\\
0&0&0&1&1&1&0&1
           \end{pmatrix}.$
\end{center}
It can be checked that $C_1\cap (\pi_2\pi_1(C_2))^\perp$ and $\pi_2\pi_1(C_2)\cap C_1^\perp$ are binary linear $[8,1]$ codes generated by the first row of $G'_1$ and $G'_2$, respectively. By Theorem \ref{thm-LlIP}, let $\pi_3\in {\rm {\bf PAut}}(\F_2^8)$ be defined by $\pi_3(c_1,c_2,c_3,c_4,c_5,c_6,c_7,c_8)=(c_1,c_2,c_3,c_5,c_4,c_6,c_7,c_8).$ Then $C_1$ and $\pi_3\pi_2\pi_1(C_2)$ respectively have generator matrices as follows:
\begin{center}
$G'''_1=\begin{pmatrix}
1&1&1&1&0&0&1&1\\
0&1&1&1&1&1&0&1\\
0&0&1&0&1&1&0&0\\
0&0&0&1&1&1&0&0
           \end{pmatrix}$ and $G'''_2=\begin{pmatrix}
1&1&1&0&1&1&0&0\\
0&1&1&0&1&1&1&0\\
0&0&1&0&1&1&1&1\\
0&0&0&1&1&1&0&1
           \end{pmatrix}.$
\end{center}
It can be checked that $\dim(C_1\cap (\pi_3\pi_2\pi_1(C_2))^\perp)=0$. Therefore, let $\sigma=\pi_3\pi_2\pi_1\in {\rm {\bf PAut}}(\F_2^8)$ be defined by $\sigma(c_1,c_2,c_3,c_4,c_5,c_6,c_7,c_8)=(c_3,c_1,c_2,c_5,c_4,c_6,c_7,c_8).$ Then we have $\dim(C_1\cap (\sigma(C_2))^\perp)=0$, that is, $(C_1,(\sigma(C_2))^\perp)$ is an LCP of codes.
\end{example}

For a linear code $C$, let $\{0\}\times C:=\{(0,{\bf c})\in \F_2\times \F_2^n~|~{\bf c}\in C\}.$

\begin{theorem}{\rm \cite[Open Problem]{Sigam-LCD}}\label{thm-binaryLCP}
Let $C_1$ and $C_2$ be binary linear $[n,k]$ codes. Then, there exists $\sigma\in {\rm {\bf PAut}}(\F_2^{n+1})$ such that $(\{0\}\times C_1,(\sigma(\{0\}\times C_2))^\perp)$ is an LCP of codes.
\end{theorem}

\begin{proof}
The proof is obvious by Theorem \ref{thm-LlIP}.
\end{proof}

\section{The best security parameters for binary LCPs of codes}

\subsection{A conjecture on the best security parameters of binary LCP of codes}
The {\em hull} of a linear code $C$ is defined by ${\rm Hull}(C):=C\cap C^{\perp}.$ This notion was introduced by Assmus and Key \cite{A-hull-DM} to classify finite projective planes.
Lemma \ref{lem-LIP-1'} is actually included in Lemma \ref{lem-LIP-1} ($C_1=C_2$). Here we provide the specific form of $\sigma \in {\rm {\bf PAut}}(\F_2^n)$. We first introduce a useful proposition.

\begin{prop}{\rm\cite[Proposition 2]{Sigam-LCD}}\label{prop-hull-rank}
Let $\sigma\in {\rm {\bf PAut}}(\F_2^n)$ and let $C$ be a binary linear $[n,k]$ code with generator matrix $G$. Then we have
$$\dim(C \cap (\sigma(C))^\perp)=k-{\rm rank}(G(\sigma(G))^T),$$
where $\sigma(G)$ is the $k\times n$ matrix with its $i$th row $\sigma(G(i,:))$, and $G(i,:)$ is the $i$th row of $G$ for $1\leq i\leq k.$
\end{prop}

\begin{lem}\label{lem-LIP-1'}
Let $C$ be a binary linear $[n,k]$ code such that $\dim({\rm Hull}(C))=h\geq 2$. Then for $1\leq \ell\leq h$, there exists $\sigma\in {\rm {\bf PAut}}(\F_2^n)$ such that
$(C,(\sigma(C))^\perp)$ is an $\ell$-intersection pairs of codes, $i.e.$, $\dim(C\cap (\sigma(C))^\perp)=\ell$.
\end{lem}

\begin{proof}
If $h=1$, then the result holds by choosing $\sigma$ as the identity map. Next, it is sufficient to consider the case where $h\geq 2$.
According to \cite[Lemma 3]{Sigam-LCD}, there exists a mapping $\pi\in {\rm {\bf PAut}}(\F_2^n)$ such that $\pi(C)$ has a generator matrix of
the form
$$G=\begin{pmatrix}
              I_{h} & A \\
              O & B
            \end{pmatrix},$$
where $(I_{h} ~ A)$ is a generator matrix of ${\rm Hull}(\pi(C))$. Hence $(I_{h} ~ A)G^T=O$ and $G(I_{h} ~ A)^T=O$. This implies that $I_h+AA^T=O$, $AB^T=O$ and $BA^T=O.$ Then it is easy to see that
\begin{align*}
  GG^T & = \begin{pmatrix}
              I_h & A \\
              O & B
            \end{pmatrix}
\begin{pmatrix}
              I_h & A \\
              O & B
            \end{pmatrix}^T\\
   & =\begin{pmatrix}
              I_h+AA^T & AB^T \\
              BA^T & BB^T
            \end{pmatrix}\\
& =\begin{pmatrix}
              O & O \\
              O & BB^T
            \end{pmatrix}.
\end{align*}
Further, by Proposition \ref{prop-hull-rank}, ${\rm rank}(GG^T)={\rm rank}(BB^T)=k-h.$ For $1\leq i\leq h$, let $\pi_i\in {\rm {\bf PAut}}(\F_2^n)$ be defined by $\pi_i(c_1,c_2,\ldots,c_i,c_{i+1},\ldots,c_n)=(c_i,c_1,\ldots,c_{i-1},c_{i+1},\ldots,c_n)$.
Let $\tau_i\in {\rm {\bf PAut}}(\F_2^h)$ be defined by $\tau_i(c_1,\ldots,c_i,c_{i+1},\ldots,c_h)=(c_i,c_1,\ldots,c_{i-1},c_{i+1},\ldots,c_h)$. Hence $\pi_i({\bf a},{\bf b})=(\tau_i({\bf a}),{\bf b})$, where ${\bf a}\in \F_2^h$ and ${\bf b}\in \F_2^{n-h}$.
Since $h\geq 2$, we have
  \begin{align*}
  G(\pi_i(G))^T & =\begin{pmatrix}
              I_h & A \\
              O & B
            \end{pmatrix}
\left(\pi_i\begin{pmatrix}
              I_h & A \\
              O & B
            \end{pmatrix}\right)^T\\
& = \begin{pmatrix}
              I_h & A \\
              O & B
            \end{pmatrix}
\begin{pmatrix}
              \tau_i(I_h) & A \\
              O & B
            \end{pmatrix}^T\\
   & =\begin{pmatrix}
              \tau_i(I_h)+AA^T & AB^T \\
              BA^T & BB^T
            \end{pmatrix}\\
& =\begin{pmatrix}
              \tau_i(I_h)-I_h & O \\
              O & BB^T
            \end{pmatrix}.
\end{align*}
Note that
$$\tau_i(I_h)-I_h  = \begin{pmatrix}
                              \begin{array}{cc}
                                {\bf 0} & 1\\
                                I_{i-1} & {\bf 0}
                              \end{array}
                               & O \\
                              O & I_{h-i}
                            \end{pmatrix}-I_h\\
    =\begin{pmatrix}
              D & O \\
              O & O
            \end{pmatrix},$$
where $$D=\begin{pmatrix}
                 -1 & 0 & 0 & \cdots & 0 & 1 \\
                 1 & -1 & 0 & \ddots & 0 & 0 \\
                 0 & 1 & -1& \ddots & 0& 0 \\
                 \vdots & \ddots & \ddots & \ddots & \ddots & \vdots \\
                 0 & 0 & 0 & \ddots & -1 & 0\\
                 0 & 0 & 0 & \cdots & 1 & -1
               \end{pmatrix}_{i\times i}.$$
It can be checked that the sum of all rows of $D$ is the zero row. Further, $D$ contains a submatrix with rank $i-1$. Hence ${\rm rank}(D)=i-1$. It follows that
$${\rm rank}(G(\pi_i(G))^T)={\rm rank}(GG^T)+{\rm rank}(D)=k-h+i-1.$$
Note that ${\rm rank}(G(\pi_{h-\ell+1}(G))^T)=k-h+(h-\ell+1)-1=k-\ell.$ Let $\sigma=\pi_{h-\ell+1}\pi\in {\rm {\bf PAut}}(\F_2^n)$.
By Proposition \ref{prop-hull-rank}, we have $\dim(C\cap (\sigma(C))^\perp)=\dim(C\cap (\pi_{h-\ell+1}\pi(C))^\perp)=k-{\rm rank}(G(\pi_{h-\ell+1}(G))^T)=k-(k-\ell)=\ell$. This completes the proof.
\end{proof}

\begin{lem}\label{lem-LIP-2}
Let $C$ be a binary linear $[n,k]$ code such that $\dim({\rm Hull}(C))=h$ and ${\bf 1}\notin {\rm Hull}(C)$. Then for $0\leq \ell\leq h$, there exists a mapping $\sigma\in {\rm {\bf PAut}}(\F_2^n)$ such that $(C,(\sigma(C))^\perp)$ is an $\ell$-intersection pairs of codes, $i.e.$, $\dim(C\cap (\sigma(C))^\perp)=\ell$.
\end{lem}

\begin{proof}
If $1\leq \ell\leq h$, then the result holds by Lemma \ref{lem-LIP-1'}. The proof of the case where $\ell=0$ is obvious by Theorem \ref{thm-LlIP}.
\end{proof}

\begin{example}
Let $C$ be a binary linear $[7,3,4]$ code with the following generator matrix
$$G=\begin{pmatrix}
            1 & 0 & 0 & 1 & 1 & 0 & 1 \\
            0 & 1 & 0 & 1 & 0 & 1 & 1 \\
            0 & 0 & 1 & 0 & 1 & 1 & 1
          \end{pmatrix}.$$
It can be checked that $C$ is self-orthogonal, $i.e.,$ $C \cap C^\perp=C.$ By Lemma \ref{lem-LIP-1'}, let $\sigma_1\in {\rm {\bf PAut}}(\F_2^7)$ be defined by $\sigma_1(c_1,c_2,c_3,c_4,c_5,c_6,c_7)=(c_3,c_1,c_2,c_4,c_5,c_6,c_7).$ Then $\sigma_1(C)$ has a generator matrix as follows:
$$\sigma_1(G)=\begin{pmatrix}
            0 &1 & 0 & 1 & 1 & 0 & 1 \\
            0 &0 & 1 & 1 & 0 & 1 & 1 \\
            1 &0 & 0 & 0 & 1 & 1 & 1
          \end{pmatrix},$$
and we have
$$G(\sigma_1(G))^T=\begin{pmatrix}
                        1 & 1& 0 \\
                        0 & 1 & 1 \\
                        1 & 0 & 1
                      \end{pmatrix}.$$
It can be checked that ${\rm rank}(G(\sigma_1(G))^T)=2$ and $C \cap (\sigma_1(C))^\perp=C^\perp \cap \sigma_1(C)=\{(0000000),(1 1 1 0 0 0 1)\}$. By Theorem \ref{thm-LlIP} and Lemma \ref{lem-LIP-2}, let $\sigma_2\in {\rm {\bf PAut}}(\F_2^7)$ be defined by $\sigma_2(c_1,c_2,c_3,c_4,c_5,c_6,c_7)=(c_1,c_2,c_4,c_3,c_5,c_6,c_7).$ Then $\sigma_2\sigma_1(C)$ has a generator matrix as follows:
$$\sigma_2\sigma_1(G)=\begin{pmatrix}
            0 &1 & 1 & 0 & 1 & 0 & 1 \\
            0 &0 & 1 & 1 & 0 & 1 & 1 \\
            1 &0 & 0 & 0 & 1 & 1 & 1
          \end{pmatrix},$$
and we have
$$G(\sigma_2\sigma_1(G))^T=\begin{pmatrix}
                        1 & 0& 0 \\
                        0 & 0 & 1 \\
                        1 & 1 & 1
                      \end{pmatrix}.$$
Then ${\rm rank}(G(\sigma_2\sigma_1(G))^T)=3$ and $\dim(C_1\cap(\sigma_2\sigma_1(C))^\perp)=0$.
Let $\sigma=\sigma_2\sigma_1\in {\rm {\bf PAut}}(\F_2^7)$ be defined by $\sigma(c_1,c_2,c_3,c_4,c_5,c_6,c_7)=(c_3,c_1,c_4,c_2,c_5,c_6,c_7)$. Then $(C,(\sigma(C))^\perp)$ is an LCP of codes.
\end{example}

Carlet $et~al.$ \cite{Sigam-LCD} showed that $d_L(n,k)-1\leq d_{LCP}(n,k)\leq d_L(n,k)$. Hence an interesting topic is to study when this lower bound is reached. For this topic, Choi $et~al.$ \cite{2023CCDS-LCP} described a sufficient condition for $d_{LCP}(n,k)= d_L(n,k)-1$, and gave a conjecture on its necessary condition as follows.

\begin{conjecture}{\rm \cite[Conjecture 10]{2023CCDS-LCP}}
Let $n>k\geq 1$. Then, there exists a unique binary optimal $[n, k, d_L(n, k)]$ code $C$ which is even-like and contains ${\bf1}$ if and only if
$$d_{LCP}(n,k)=d_{L}(n,k)-1.$$
\end{conjecture}

In fact, we see in the subsequent proof that $d_{LCP}(n,k)=d_{L}(n,k)-1$ if and only if every binary optimal $[n, k, d_L(n, k)]$ code is even-like and contains ${\bf1}$.

\begin{theorem}\label{thm-2optimalLCP}
Let $n>k\geq 1$. Then, $d_{LCP}(n,k)=d_{L}(n,k)-1$ if and only if every binary optimal $[n, k, d_L(n, k)]$ code is even-like and contains ${\bf1}$.
\end{theorem}

\begin{proof}
$(\Longleftarrow)$ Suppose that $d_{LCP}(n,k)=d_{L}(n,k)$, then there exists a binary $[n,k]$ LCP of codes $(C,D)$ with security parameter $d_L(n, k)$. Hence $C\cap D=\{{\bf 0}\}$, $C$ and $D^\perp$ have parameters $[n,k,d_L(n,k)]$.
By hypothesis, every binary optimal $[n, k, d_L(n, k)]$ code is even-like and contains ${\bf1}$. This implies that both $C$ and $D^\perp$ are even-like and contain ${\bf1}$. Since $D^\perp$ is even-like, we have ${\bf 1}\in D$. Then ${\bf 1}\in C\cap D$. This contradicts the condition that $C\cap D=\{{\bf 0}\}$. Hence, $d_{LCP}(n,k)\leq d_{L}(n,k)-1$. On the other hand, by \cite[Corollary 9]{Sigam-LCD}, we have $d_{LCP}(n,k)\geq d_{L}(n,k)-1$. Therefore, $d_{LCP}(n,k)= d_{L}(n,k)-1$.

$(\Longrightarrow)$ Suppose that $d_{LCP}(n,k)=d_{L}(n,k)-1$. If there exists a binary linear $[n, k, d_L(n, k)]$ code $C$ such that $C$ is odd-like or does not contain ${\bf1}$, then both cases imply that ${\bf 1}\notin {\rm Hull}(C)$. By Lemma \ref{lem-LIP-2}, there exists a mapping $\sigma\in {\rm {\bf PAut}}(\F_2^n)$ such that $(C,(\sigma(C))^\perp)$ is an LCP of codes. This implies that $d_{LCP}(n,k)=d_{L}(n,k)$, which contradicts the condition that $d_{LCP}(n,k)=d_{L}(n,k)-1$.
\end{proof}

\begin{cor}\label{cor-LCP}
If there exists a binary linear $[n,k,d]$ code $C$ with ${\bf1}\notin {\rm Hull}(C)$, then there exists a binary $[n,k]$ LCP of codes with security parameter $d$.
\end{cor}

\begin{proof}
By Lemma \ref{lem-LIP-2}, there exists a mapping $\sigma\in {\rm {\bf PAut}}(\F_2^n)$ such that $(C,(\sigma(C))^\perp)$ is a binary $[n,k]$ LCP of codes. The security parameter of $(C,(\sigma(C))^\perp)$ is equal to $\min\{d(C),d(\sigma(C))\}$. Together with the fact that $d(\sigma(C))=d(C)$, the result holds.
\end{proof}

\subsection{Binary LCPs of codes from Solomon-Stiffler codes}
Very recently, G${\rm\ddot{u}}$neri \cite{2023ITSSLCP} constructed an infinite family of optimal binary LCPs of codes from Solomon-Stiffler codes. Here, we construct an infinite family of optimal binary LCPs of codes from Solomon-Stiffler codes, whose parameters cover those constructed by G${\rm\ddot{u}}$neri \cite{2023ITSSLCP}. We first recall the concept of Solomon-Stiffler codes.

Assume that $S_k$ is a $k\times (2^k-1)$ matrix whose columns are made up of all nonzero vectors of $\F_2^k$.
Let $sS_k$ denote $s$ copies of a matrix $S_k$, in other words,
$$sS_k=(\underbrace{S_k~|~S_k~|~\cdots~|~S_k}_s).$$

If $U$ and $V$ are matrices with the same number of rows and $V$ is a submatrix of $U$, then we denote by $(U\setminus V)$ the matrix obtained by puncturing the columns of $V$ from $U$. Let $U$ be a $u$-dimensional subspace of $\F_2^k$. We denote the $k \times (2^u -1)$ matrix of rank $u$, whose columns consist of nonzero elements of the space $U$, also by $U$. It is clear that the matrix $U$ is contained in $S_k$ for $1\leq u< k$.

Let $G$ be a generator matrix of a binary $[n,k,d]$ linear code $C$ with $d(C^\perp)\geq 2$.
Then there exists an integer $s$ such that $G$ is a submatrix of $sS_k$. We define a $k\times (s(2^k-1)-n)$ matrix $G'$ by
\begin{align}\label{eq-anticode}
 G'=(sS_k \setminus G).
\end{align}
Then the rows of the matrix $G'$ generate a code $C'$, which is called the anticode of $C$, namely,
$C'=\{{\bf c}'~|~{\bf c}'={\bf x}G',~{\bf x}\in \F_2^k\}.$
For ${\bf c}={\bf x}G\in C$, we define ${\bf c'}={\bf x}G'\in C'$. Since the rows of $sS_k$ generate a one-weight $[s(2^k-1),k,s2^{k-1}]$ code, by (\ref{eq-anticode}) we have
\begin{align}\label{eq-1}
  {\rm wt}({\bf c})+{\rm wt}({\bf c'})=s2^{k-1},~{\bf c}\in C\backslash \{{\bf 0}\}.
\end{align}

Given $k$ and $d$, we define $s=\left\lceil\frac{d}{2^{k-1}}\right\rceil$. Then $s2^{k-1}-d< 2^{k-1}$ and
there exist positive integers $u_1,u_2,\ldots,u_p$ such that $k>u_1> u_2>\cdots >u_p\geq 1$ and
 $$s2^{k-1}-d=\sum_{i=1}^p2^{u_i-1}.$$

We define
$\mathcal{U}(k,u)=\left\{U~|~U=\widehat{U}\backslash \{{\bf 0}\},~\widehat{U}~{\rm is~a}~u\right.$-dimensional subspace of $\left.\F_2^k\right\}.$
In 1965, Solomon and Stiffler~\cite{SS-code} found a family of Griesmer codes by specifying $G'$ as follows:
\begin{align*}
  G'=\left(U_1~|~U_2~|~\cdots~|~U_p\right),
\end{align*}
where $U_i\in \mathcal{U}(k,u_i)$, $k>u_1> u_2>\cdots >u_p\geq 1$, and $U_i\cap U_j=\emptyset$ when $i\neq j$.
Note that the rows of the matrix $G'$ generate a binary code $C'$ with the maximum weight $\sum_{i=1}^p 2^{u_i-1}$. By (\ref{eq-1}) and the Griesmer bound,
the code $C$ with the anticode $C'$ is a binary linear code with the parameters
\begin{align}\label{eq-2}
\left[s(2^k-1)-\sum_{i=1}^p(2^{u_i}-1),k,s2^{k-1}-\sum_{i=1}^p2^{u_i-1}\right].
\end{align}
Then $C$ is called a binary Solomon-Stiffler code~\cite{SS-code}.
Belov \cite{Belov-1} showed that if $\sum_{i=1}^{\min\{s+1, p\}}u_i\leq sk$, then we can choose $U_i\in \mathcal{U}(k,u_i)~(1\leq i\leq p)$ appropriately to ensure the existence of an Solomon-Stiffler code with the parameters $(\ref{eq-2})$.

\begin{cor}\label{cor-SS-LCP}
Assume that $s\geq 2$ and $k\geq 2$. If $\sum_{i=1}^{\min\{s+1, p\}}u_i\leq sk$, then there exists a binary $[g(k,d),k]$ LCP of codes with security parameter $d=s2^{k-1}-\sum_{i=1}^p2^{u_i-1}$.
\end{cor}

\begin{proof}
It is well-known that the rows of the matrix $sS_k$ generate a one-weight $[s(2^k-1),k,s2^{k-1}]$ Griesmer code $C$. Hence we can obtain a punctured code of $C$ which contains ${\bf 1}$ by removing at least $s(2^k-1)-s2^{k-1}=s2^{k-1}-s$ positions. Note that the Solomon-Stiffler code with the parameters $(\ref{eq-2})$ is a punctured code of $C$ by puncturing at $\sum_{i=1}^p(2^{u_i}-1)$ positions. However,
$$\sum_{i=1}^p(2^{u_i}-1)\leq \sum_{j=1}^{k-1}(2^{j}-1)=2^k-k-1\leq s2^{k-1}-s.$$
Hence the Solomon-Stiffler code with the parameters $(\ref{eq-2})$ does not contain ${\bf 1}$. By Corollary \ref{cor-LCP}, the result holds.
\end{proof}

From this, we can obtain the main result of \cite{2023ITSSLCP}.

\begin{cor}{\rm\cite[Theorem 4]{2023ITSSLCP}}
For any $k\geq 5$ and $d\geq \left\lceil \frac{k-1}{2}\right\rceil2^{k-1}$, there exists a binary $[g(k,d),k]$ LCP of codes with security parameter $d$.
\end{cor}

\begin{proof}
By \cite{1973ITSS}, there exists a binary Solomon-Stiffler $[g(k,d),k]$ code for any $k\geq 5$ and $d\geq \left\lceil \frac{k-1}{2}\right\rceil2^{k-1}$.
Note that
$$s=\left\lceil\frac{d}{2^{k-1}}\right\rceil\geq \left\lceil \frac{k-1}{2}\right\rceil \geq 2.$$
By Corollary \ref{cor-SS-LCP}, the result holds.
\end{proof}

\begin{cor}\label{cor-LCP-1}
Suppose that $n=2^k$, where $k\geq 1$. Then we have $d_{LCP}(n,k+1)=d_{L}(n,k+1)-1=2^{k-1}-1.$
\end{cor}

\begin{proof}
From \cite{Tor}, any binary linear $[2^k,k+1,2^{k-1}]$ code $C$ is equivalent to the first order Reed-Muller code $\mathcal{R}(1,k)$. Hence $C$ contains the all-one vector ${\bf 1}$. By Theorem \ref{thm-2optimalLCP}, $d_{LCP}(n,k+1)=d_{L}(n,k+1)-1=2^{k-1}-1.$
\end{proof}

Guenda $et~al.$ \cite{2020DCC-lIP} introduced the concept of linear $\ell$-intersection pairs of codes, and proposed a conjecture on the existence of linear $\ell$-intersection pairs of codes as follows.

\begin{conjecture}{\rm \cite[Conjecture 2.1]{2020DCC-lIP}}
There exists a linear $\ell$-intersection pair of $[n,k_1,d_1]_q$ and $[n,k_2,d_2]_q$ codes for every $\ell$ satisfying $k_1+k_2-n\leq \ell\leq \min\{k_1,k_2\}$ provided that there exist $[n,k_1,d_1]_q$ and $[n,k_2,d_2]_q$ codes.
\end{conjecture}

\begin{remark}
Guenda $et~al.$ \cite{2020DCC-lIP} and Huang $et~al.$ \cite{HFF2022l-intesection222} showed that this conjecture holds for maximum distance separable (MDS) codes over $\F_q$ of length less than or equal to $q +1$, where $q>2$. In fact, Corollary \ref{cor-LCP-1} implies that this conjecture does not always hold for the binary case.
\end{remark}

\section{Conclusion}
In this paper, we have pushed further the study of the characterization of binary LCPs of codes. We have characterized the existence of binary LCPs of codes. As a consequence, an open problem proposed by Carlet {\em et al.} in \cite{Sigam-LCD} has been solved. We have also settled the best parameter problem of binary LCPs of codes. As a result, a conjecture proposed by Choi {\em et al.} in \cite{2023CCDS-LCP} has been solved. Further, we have constructed an infinite family of optimal binary LCPs of codes with more general parameters, which cover the parameters constructed by G${\rm\ddot{u}}$neri in \cite{2023ITSSLCP}. Further, these results also imply that the conjecture proposed by Guenda $et~al.$ in \cite{2020DCC-lIP} does not always hold for the binary case.\vspace{3mm}

\noindent{\bf Acknowledgement:} The research of Shitao Li and Minjia Shi is supported by the National Natural Science Foundation of China under Grant 12071001. The research of San Ling is supported by Nanyang Technological University Research Grant 04INS000047C230GRT01.



\end{sloppypar}

\end{document}